\documentclass[11pt]{article}

\skip\footins=4ex                                
\usepackage{amsthm}
\usepackage{amsmath}
\usepackage{amssymb}

\usepackage[ps2pdf, letterpaper, pdfpagemode=UseNone]{hyperref} 

\usepackage{graphicx}

\usepackage{algorithmic}
\usepackage{algorithm}

\newtheorem{theorem}{Theorem}
\newtheorem{lemma}[theorem]{Lemma}

\newtheorem{definition}[theorem]{Definition}

\newcommand{\R}{\mathbb{R}}                     
\newcommand{\F}{\mathbb{F}}                     

\newcommand{\subsets}[2]{\binom{#1}{#2}} 
\newcommand{\E}[1]{\mathbf{E}\left[#1\right]}
\newcommand{\pr}[1]{\mathrm{Pr}\left[#1\right]}
\newcommand{\setof}[2]{\{#1~:~#2\}}
\newcommand{\set}[1]{\{#1\}}

\newcommand{\oneone}{1\!\!1}
\newcommand{\one}[1]{\oneone\left(#1\right)}
\newcommand{\piecewise}[1]{\left\{\begin{array}{@{\;}ll}#1\end{array}\right.}
\newcommand{\tif}{\textrm{If }}
\newcommand{\ow}{\textrm{Otherwise}}

\newcommand{\svk}{\subsets{V}{k}}
\newcommand{\svkmo}{\subsets{V}{k-1}}

\newcommand{\bnkmo}{\binom{n}{k-1}} 
\newcommand{\bnk}{\binom{n}{k}}
\newcommand{\bnkmt}{\binom{n}{k-2}} 
\newcommand{\sm}{\setminus} 

\newcommand{\calC}{\mathcal{C}}
\newcommand{\ceil}[1]{\lceil #1 \rceil}
\newcommand{\bigO}[1]{O\left(#1\right)}

\newcommand{\Obj}{Obj} 

\newcommand{\footcomment}[1]{} 
\newcommand{\margincomment}[1]{} 
\newcommand{\commentw}[1]{\margincomment{WS: #1}}

\begin{document}
\title{Linear Time Approximation Schemes for the Gale-Berlekamp Game and Related Minimization Problems}
\author{Marek Karpinski\footnote{ marek@cs.uni-bonn.de. Dept. of Computer Science, University of Bonn. Part of this work was done while visiting Microsoft Research.}{ } and Warren Schudy\footnote{ws@cs.brown.edu. Dept. of Computer Science, Brown University. Part of this work was done while visiting University of Bonn.}}
\date{}
\maketitle

\begin{abstract}
We design a linear time approximation scheme for the Gale-Berlekamp Switching Game and generalize it to a wider class of dense fragile minimization problems including the Nearest Codeword Problem (NCP) and Unique Games Problem. Further applications include, among other things, finding a constrained form of matrix rigidity and maximum likelihood decoding of an error correcting code. As another application of our method we give the first linear time approximation schemes for correlation clustering with a fixed number of clusters and its hierarchical generalization. Our results depend on a new technique for dealing with small objective function values of optimization problems and could be of independent interest.
\end{abstract}

\thispagestyle{empty} 
\setcounter{page}{0} 

\clearpage

\section{Introduction}\label{sec:intro}

The Gale-Berlekamp Switching Game (GB Game) was introduced independently by Elwyn Berlekamp \cite{CS04, Spe94} and David Gale~\cite{Spe94} in the context of coding theory. This game is played using of a $m$ by $m$ grid of lightbulbs. The adversary chooses an arbitrary subset of the lightbulbs to be initially ``on.'' Next to every row (resp. column) of lightbulbs is a switch, which can be used to invert the state of every lightbulb in that row (resp. column). The protagonist's task is to minimize the number of lit lightbulbs (by flipping switches). This problem was proven very recently to be NP-hard~\cite{RV08}. Let $\Phi = \set{-1,1} \subset \R$. For matrices $M,N$ let $d(M,N)$ denote the number of entries where $M$ and $N$ differ. It is fairly easy to see that the GB Game is equivalent to the following natural problems:~\cite{RV08}
\begin{itemize}
\item Given matrix $M \in \Phi^{m \times m}$ find row vectors $x, y \in \Phi^m$ minimizing $d(M, x y^T)$.
\item Given matrix $M \in \Phi^{m \times m}$ find rank-1 matrix $N \in \Phi^{m \times m}$ minimizing $d(M, N)$.
\item Given matrix $M \in \F_2^{m \times m}$ find $x, y \in \F_2^m$ minimizing $\sum_{ij} \one{M_{ij} \neq x_i \oplus y_j}$ where $\F_2$ is the finite field over two elements with addition operator $\oplus$.
\item Given matrix $M \in \Phi^{m \times m}$ find row vectors $x, y \in \Phi^m$ \emph{maximizing} $x^T M y$.
\end{itemize}

We focus on the equivalent minimization versions and prove existence of linear-time approximation schemes for them.

\begin{theorem}\label{thm:GB}
For every $\epsilon>0$ there is a randomized $1+\epsilon$-approximation algorithm for the Gale-Berlekamp Switching Game (its minimization version) with runtime $O(m^2) + 2^{O(1/\epsilon^2)}$.
\end{theorem}

In order to achieve the linear-time bound of our algorithms, we introduce two new techniques: calling the additive error approximation algorithm \emph{at the end} of our algorithm and greedily refining the random sample used by the algorithm. These new methods could also be of independent interest.

A constraint satisfaction problem (CSP) consists of $n$ variables over a domain of constant-size $d$ and a collection of arity-$k$ constraints ($k$ constant). The objective of MIN-$k$CSP (MAX-$k$CSP) is to minimize the number of unsatisfied (maximize the number of satisfied) constraints. An (everywhere) dense instance is one where every variable is involved in at least a constant times the maximum possible number of constraints, i.e. $\Omega(n^{k-1})$. For example, the GB Game is a dense MIN-2CSP since each of the $n=2m$ variables is involved in precisely $m=n/2$ constraints. It is natural to consider generalizing Theorem~\ref{thm:GB} to all dense MIN-CSPs, but unfortunately many such problems have no PTASs unless P=NP~\cite{BFK03} so we must look at a restricted class of MIN-CSPs. A constraint is \emph{fragile} if modifying any variable in a satisfied constraint makes the constraint unsatisfied. A CSP is \emph{fragile} if all of its constraints are. Clearly the GB Game can be modeled as a fragile dense MIN-2CSP. Our results generalize to all dense \emph{fragile} MIN-$k$CSPs.

We now formulate our general theorem.

\begin{theorem}\label{thm:minCSP}\label{thm:simple}
For every $\epsilon>0$ there is a randomized $1+\epsilon$-approximation algorithm for dense fragile MIN-kCSPs with runtime $O(n^k) + 2^{O(1/\epsilon^2)}$.
\end{theorem}

Any approximation algorithm for MIN-$k$CSP must read (by adversary argument) the entire input to distinguish between instances with optimal value of 1 and 0 and hence the $O(n^k)$ term of the runtime cannot be improved. It is fairly easy to see that improving the second term (to $2^{o(1/\epsilon^2)}$) would imply a $O(n^2) + 2^{o(1/\epsilon^2)}$-time PTAS for average-dense max cut. Over a decade worth of algorithms~\cite{AFK96, AKK95, FK99, AFKK02, MS08} for MAX-kCSP all have dependence on $\epsilon$ of at best $2^{O(1/\epsilon^2)}$, so any improvement to the runtime of Theorem~\ref{thm:minCSP} would be surprising.

We begin exploring applications of Theorem~\ref{thm:minCSP} by generalizing the Gale-Berlekamp game to higher dimensions $k$ ($k$-ary GB) and then to arbitrary $k$-ary equations. Given $n$ variables $x_i \in \F_2$ and $m$ linear equations of the form $x_{i_1} \oplus x_{i_2} \oplus \ldots \oplus x_{i_k} = 0$ (or $=1$), the $k$-ary Nearest Codeword Problem (NCP) consists of finding an assignment minimizing the number of unsatisfied equations. As the name suggests, the Nearest Codeword Problem can be interpreted as maximum likelihood decoding for linear error correcting codes. The Nearest Codeword Problem has fragile constraints so Theorem~\ref{thm:minCSP} implies a linear-time PTAS for the $k$-ary GB problem and the dense $k$-ary Nearest Codeword Problem.

The Unique Games Problem (UGP)~\cite{FL92,GT06} consists of solving MIN-2CSPs where the constraints are permutations over a finite domain $D$ of colors; i.e. a constraint involving variables $x_u$ and $x_v$ is satisfied iff $x_u = \pi_{uv}(x_v)$ for permutation $\pi_{uv}$. These constraints are clearly fragile, so Theorem~\ref{thm:minCSP} implies also a linear-time PTAS for the dense Unique Game Problem (with a constant number of colors).

The multiway cut problem, also known as MIN-$d$CUT, consists of coloring an undirected graph with $d$ colors, such that each of $d$ terminal nodes $t_i$ is colored with color $i$, minimizing the number of bichromatic edges. The requirement that the terminal nodes must be colored particular colors does not fit in our dense fragile MIN-CSP framework, so we use a work-around: let the constraint corresponding to an edge be satisfied only if it is monochromatic \emph{and} the endpoint(s) that are terminals (if any) are colored correctly.

As another application, consider MIN-$k$SAT, the problem of minimizing the number of \emph{satisfied} clauses of a boolean expression in conjunctive normal form where each clause has $k$ variables (some negated). We consider the equivalent problem of minimizing the number of \emph{unsatisfied} conjunctions of a boolean expression in \emph{disjunctive} normal form. A conjunction can be represented as a fragile constraint indicating that all of the negated variables within that constraint are false and the remainder are true, so Theorem~\ref{thm:minCSP} applies to MIN-$k$SAT as well.

Finally we consider correlation clustering and hierarchical clustering with a fixed number of clusters \cite{GG06,AC05}. Correlation cluster consists of coloring an undirected graph with $d$ colors (like multiway cut so far), minimizing the sum of the number of cut edges and the number of uncut non-edges.  Correlation clustering with two clusters is equivalent to the following symmetric variant of the Gale-Berlekamp game: given a symmetric matrix $M \in \Phi^{m \times m}$ find a row vector $x \in \Phi^m$ minimizing $d(M, x x^T)$. Like the GB game, correlation clustering with 2 clusters is fragile and Theorem~\ref{thm:simple} gives a linear-time approximation scheme. For $d>2$ correlation clustering is not fragile but has properties allowing for a PTAS anyway. We also solve a generalization of correlation clustering called hierarchical clustering \cite{AC05}. We prove the following theorem.

\begin{theorem}\label{thm:cc}
For every $\epsilon>0$ there is a randomized $1+\epsilon$-approximation algorithm for correlation clustering and hierarchical clustering with fixed number of clusters $d$ with running time $n^2 2^{O(d^6/\epsilon^2)}$.
\end{theorem}

The above results improves on the running time $O(n^{9^d / \epsilon^2})\log n = O(n^{9^d / \epsilon^2})$ of the previous PTAS for correlation clustering by Giotis and Guruswami \cite{GG06} in two ways: first the polynomial is linear in the size of the input and second the exponent is polynomial in $d$ rather than exponential. Our result for hierarchical clustering with a fixed number of clusters is the first PTAS for that problem.

\bigskip

We prove Theorem \ref{thm:simple} in Sections~\ref{sec:FD} and~\ref{sec:FDAnalysis} and Theorem~\ref{thm:cc} in Sections~\ref{sec:cc} and~\ref{sec:ccAnalysis}.

\subsection*{Related Work}

Elwyn Berlekamp built a physical model of the GB game with either $m=8$ or $m=10$~\cite{CS04, Spe94} at Bell Labs in the 1960s motivated by the connection with coding theory and the Nearest Codeword Problem. Several works~\cite{FS89, CS04} investigated the cost of worst-case instances of the GB Game; for example the worst-case instance for $m=10$ has cost 35~\cite{CS04}. Roth and Viswanathan~\cite{RV08} showed very recently that the GB game is in fact NP-hard. They also give a linear-time algorithm if the input is generated by adding random noise to a cost zero instance. Replacing $\Phi$ with $\R$ in the third formulation of the GB Game yields the problem of computing the 1-rigidity of a matrix. Lower bounds on matrix rigidity have applications to circuit and communication complexity~\cite{Lok95}.

The Nearest Codeword Problem is hard to approximate in general~\cite{ABSS93, DKRS03} better than $n^{\Omega(1/\log \log n)}$. It is hard even if each equation has exactly 3 variables and each variable appears in exactly 3 equations \cite{BK02b}. There is a $O(n / \log n)$ approximation algorithm \cite{BK02a, APY08}.

Over a decade ago two groups~\cite{AKK95, F96} independently discovered polynomial-time approximation algorithms for MAX-CUT achieving additive error of $\epsilon n^2$, implying a PTAS for average-dense MAX-CUT instances. The fastest algorithms~\cite{AFKK02,MS08} have constant runtime $2^{O(1/\epsilon^2)}$ for approximating the value of any MAX-$k$CSP over a binary domain $D$. This can be generalized to an arbitrary domain $D$. To see this, note that we can code $D$ in binary and correspondingly enlarge the arity of the constraints to $k \ceil{\log |D|}$. A random sample of $\tilde O(1/\epsilon^4)$ variables suffices to achieve an additive approximation~\cite{AFKK02, MS08, RV07}.  These results extend to MAX-BISECTION~\cite{FKK06}.

Arora, Karger and Karpinski~\cite{AKK95} introduced the first PTASs for dense \emph{minimum} constraint satisfaction problems. They give PTASs with runtime $n^{O(1/\epsilon^2)}$~\cite{AKK95} for min bisection and multiway cut (MIN-d-CUT). Bazgan, Fernandez de la Vega and Karpinski~\cite{BFK03} designed PTASs for MIN-SAT and the nearest codeword problem with runtime $n^{O(1/\epsilon^2)}$. Giotis and Guruswami \cite{GG06} give a PTAS for correlation clustering with $d$ clusters with runtime $O(n^{9^d / \epsilon^2})$. We give linear-time approximation schemes for all of the problems mentioned in this paragraph except for the MIN-BISECTION problem.

\section{Fragile-dense Algorithm}\label{sec:FD}
\subsection{Intuition}\label{sec:intuition}

Consider the following scenario. Suppose that our nemesis, who knows the optimal solution to the Gale-Berlekamp problem shown in Figure~\ref{fig:gb}, gives us a constant size random sample of it to tease us. How can we use this information to construct a good solution? One reasonable strategy is to set each variable greedily based on the random sample. Throughout this section we will focus on the row variables; the column variables are analogous. For simplicity our example has the optimal solution consisting of all of the switches in one position, which we denote by $\alpha$. For row $v$, the greedy strategy, resulting in assignment $x^{(1)}$, is to set switch $v$ to $\alpha$ iff $\hat b(v,\alpha) < \hat b(v,\beta)$, where $\hat b(v,\alpha)$ (resp. $\hat b(v,\beta)$) denotes the number of light bulbs in the intersection of row $v$ and the sampled columns that would be lit if we set the switch to position $\alpha$ (resp. $\beta$).

With a constant size sample we can expect to set most of the switches correctly but a constant fraction of them will elude us. Can we do better? Yes, we simply do greedy again. The greedy prices analogous to $\hat b$ are shown in the columns labeled with $b$ in the middle of Figure~\ref{fig:gb}.  For the example at hand, this strategy works wonderfully, resulting in us reconstructing the optimal solution exactly, as evidenced by the $b(x^{(1)},v,\alpha)<b(x^{(1)},v,\beta)$ for all $v$. In general this does not reconstruct the optimal solution but provably gives something close.

Some of the rows, e.g. the last one, have $b(x^{(1)},v,\alpha)$ much less than $b(x^{(1)},v,\beta)$ while other rows, such as the first, have $b(x^{(1)},v,\alpha)$ and $b(x^{(1)},v,\beta)$ closer together. We call variables with $|b(x^{(1)},v,\alpha) - b(x^{(1)},v,\beta)| > \Theta(n)$ \emph{clearcut}. Intuitively, one would expect the clearcut rows to be more likely correct than the nearly tied ones. In fact, we can show that we get all of the clearcut ones correct, so the remaining problem is to choose values for the rows that are close to tied. However, those rows have a lot of lightbulbs lit, suggesting that the optimal value is large, so it is reasonable to run an additive approximation algorithm and use that to set the remaining variables.

Finally observe that we can simulate the random sample given by the nemesis by simply taking a random sample of the variables and then doing exhaustive search of all possibly assignments of those variables. We have just sketched our algorithm.

Our techniques differ from previous work~\cite{BFK03, AKK95, GG06} in two key ways:
\begin{enumerate}
\item Previous work used a sample size of $O((\log n)/\epsilon^2)$, which allowed the clearcut variables to be set correctly after a single greedy step. We instead use a constant-sized sample and run a second greedy step before identifying the clearcut variables.
\item Our algorithm is the first one that runs the additive error algorithm after identifying clearcut variables. Previous work ran the additive error algorithm at the beginning.
\end{enumerate}

The same ideas apply to all dense fragile CSPs. In the remainder of the paper we do not explicitly discuss the GB Game but present our ideas in the abstract framework of fragile-dense CSPs.

\begin{figure}[tb]
\begin{center}
\includegraphics[width=\textwidth]{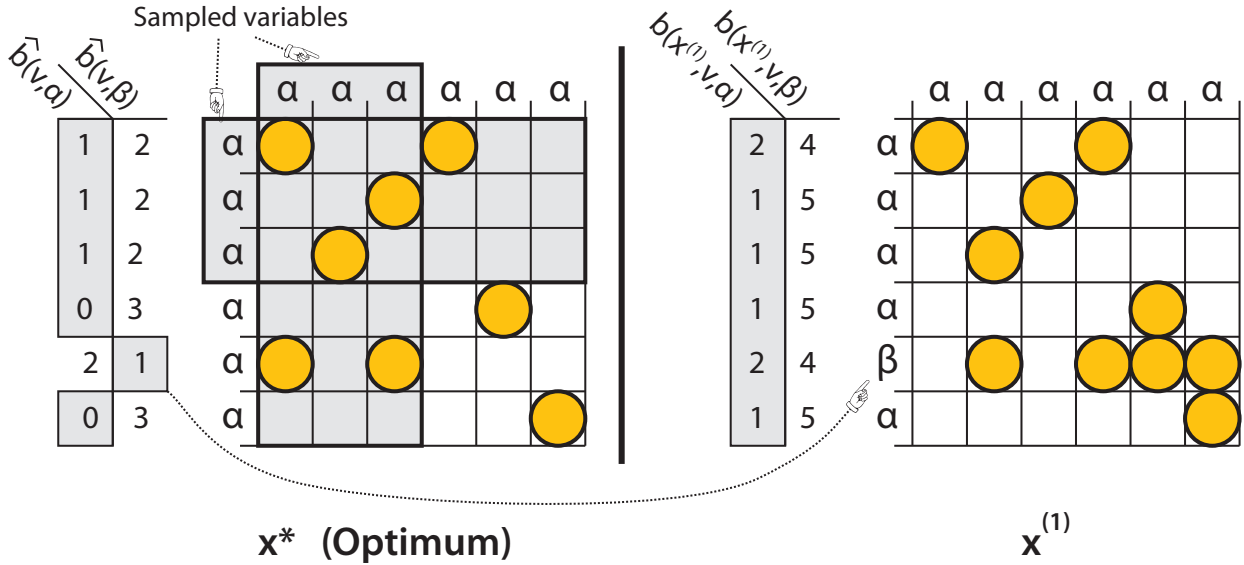}
\end{center}
\caption{An illustration of our algorithmic ideas on the Gale-Berlekamp Game.}
\label{fig:gb}
\end{figure}

\subsection{Model}
We now give a formulation of MIN-$k$CSP that is suitable for our purposes.
For non-negative integers $n, k$, let $\binom{n}{k}=\frac{n!}{k!(n-k)!}$,
 and for a given set $V$ let $\subsets{V}{k}$ denote the set of subsets of $V$ of size $k$ (analogous to $2^S$ for all subsets of $S$).
 There is a set $V$ of $n$ variables, each of which can take any value in constant-sized domain $D$.\commentw{For polytime, $|D|=2^{\sqrt{\log n}}$ is ok}{} Let $x_v \in D$ denote the value of variable $v$ in the assignment $x$.

Consider some $I \in \svk$. There may be many constraints over these variables; number them arbitrarily. Define $p(I, \ell, x)$ to be 1 if the $\ell$th constraint over $I$ is unsatisfied in assignment $x$ and zero otherwise. For $I \in \subsets{V}{k}$, we define $p_I(x)= \frac{1}{\eta} \sum_\ell p(I, \ell, x)$, where $\eta$ is a scaling factor to ensure $0 \leq p_I(x) \leq 1$ (e.g. $\eta=2^k$ for MIN-$k$SAT). For notational simplicity we write $p_I$ as a function of a complete assignment, but $p_I(x)$ only depends on $x_u$ for variables $u \in I$. For $I \not \in \subsets{V}{k}$ define $p_I(x) = 0$. 

\begin{definition}
On input $V, p$ a \emph{minimum constraint satisfaction problem} (MIN-$k$CSP) is a problem of finding an assignment $x$ minimizing $\Obj(x)=\sum_{I \in \svk} p_I(x)$.
\end{definition}

Let $R_{vi}(x)$  be an assignment over the variables $V$ that agrees with $x$ for all $u\in V$ except for $v$ where it is $i$; i.e. $R_{vi}(x)_u = \left\{\begin{array}{ll} i & \textrm{if }u=v \\ x_u & \textrm{otherwise} \end{array} \right.$. We will frequently use the identity $R_{vx_v}(x) = x$. Let $b(x, v, i) = \sum_{I \in \svk : v \in I} p_I(R_{vi}(x))$ be the number of unsatisfied constraints $v$ would be in if $x_v$ were set to $i$ (divided by $\eta$).

We say the $\ell$th constraint over $I$ is \emph{fragile} if $p(I, \ell, R_{vi}(x)) + p(I,\ell, (R_{vj}(x)) \geq 1$ for all $v \in I$ and $i \neq j \in D$.

\begin{definition}
A Min-$k$CSP is \emph{fragile-dense} if $b(x, v, i) + b(x,v,j) \geq \delta \bnkmo$ for some constant $\delta>0$ and for all assignments $x$, variables $v$ and distinct values $i$ and $j$.
\end{definition}

\begin{lemma}
An instance where every variable $v \in V$ participates in at least $\delta \eta \binom{n}{k-1}$ fragile constraints for some constant $\delta>0$
 is fragile-dense (with the same $\delta$).
\end{lemma}

\begin{proof}
By definitions:
\begin{eqnarray*}
b(x,v,i) + b(x,v,j) &=& \sum_{I\in \svk:v\in I} \left(p_I(R_{vi}(x)) + p_I(R_{vj}(x))\right) \\
& = & \sum_{I\in \svk:v\in I} \frac{1}{\eta} \sum_{\ell} \left(p(I, \ell, R_{vi}(x)) + p(I, \ell, R_{vj}(x))\right) \\
&\geq& \sum_{I\in \svk:v\in I} \frac{1}{\eta} \cdot (\textrm{The number of fragile constraints over }I) \\
& \geq & \frac{\delta \eta}{\eta} \bnkmo = \delta \bnkmo
\end{eqnarray*}
\end{proof}

We will make no further mention of individual constraints, $\eta$ or fragility; our algorithms and analysis use $p_I$ and the fragile-dense property exclusively.

\subsection{Algorithm}

We now describe our linear-time algorithms. The
main ingredients of the algorithm are new iterative applications of additive
error algorithms and a special greedy technique for refining random samples
of constant size.

\begin{algorithm}[tb]
\caption{Our algorithm for dense-fragile MIN-$k$CSP}
\label{alg:main2}
\begin{algorithmic}[1] 
\STATE Run a $\frac{\epsilon}{1+\epsilon}\delta^2 / 72 k\bnk$ additive approximation algorithm.
\IF {$Obj(answer) \geq \bnk \delta^2 / 72 k$}
 \STATE Return $answer$.
\ELSE
 \STATE Let $s = \frac{18 \log (480 |D| k / \delta)}{\delta^2}$
 \STATE Draw $S_1, S_2,\ldots,S_s$ randomly from $\svkmo$ with replacement.
 \FOR{Each assignment $\hat{x^*}$ of the variables in $\bigcup_{j=1}^s S_j$}
  \STATE For all $v$ and $i$ let $\hat b(v,i) = \frac{\binom{n}{k-1}}{s} \sum_{j = 1}^s p_{S_j\cup \{v\}} (R_{vi}(\hat{x^*}))$
  \STATE For all $v \in V$ let $x^{(1)}_v = \arg \min_i \hat b(v,i)$
  \STATE For all $v \in V$ let $x^{(2)}_v = \arg \min_i b(x^{(1)}, v, i)$
  \STATE Let $C=\setof{v \in V}{b(x^{(1)}, v,x^{(2)}_v) < b(x^{(1)},v,j) - \delta \bnkmo / 6 \textrm{ for all }j \neq x^{(2)}_v}$.
  \STATE Find $x^{(3)}$ of cost at most $\frac{\epsilon |V \sm C| \delta}{3 n} \bnkmo + \min \left[\Obj(x)\right]$ using an additive approximation algorithm, where the minimum ranges over $x$ such that $x_v = x^{(2)}_v \,\,\forall v \in C$.
 \ENDFOR
 \STATE Return the best assignment $x^{(3)}$ found.
\ENDIF
\end{algorithmic}
\end{algorithm}

Let $s=\frac{18 \log (480 |D| k / \delta)}{\delta^2}$ and $S_1,S_2,\ldots,S_s$ be a multiset of independent random samples of $k-1$ variables from $V$.
One can estimate $b(x^*,v,i)$ using the unbiased estimator $\hat b(v,i) = \frac{\binom{n}{k-1}}{s} \sum_{j = 1}^s p_{S_j\cup \{v\}} (R_{vi}(\hat{x^*}))$ (see Lemma~\ref{lem:fewClearCorrupted} for proof). One can determine the necessary $x_v^*$ by exhaustively trying each possible combination.

\section{Analysis of Algorithm~\ref{alg:main2}}\label{sec:FDAnalysis}

We use one of the known additive error approximation algorithms for MAX-$k$CSP problems.
\begin{theorem}\label{thm:MS08} \cite{MS08}
For any MAX-$k$CSP (or MIN-$k$CSP) and any $\epsilon'>0$ there is a randomized algorithm which returns an assignment of cost at most $OPT + \epsilon' n^k$ in runtime $O(n^k) + 2^{O(1/\epsilon'^2)}$.
\end{theorem}

Throughout the rest of the paper let $x^*$ denote an optimal assignment.

First consider Algorithm~\ref{alg:main2} when the ``then'' branch of the ``if'' is taken. Choose constants appropriately so that the additive error algorithm fails with probability at most $1/10$ and assume it succeeds. Let $x^a$ denote the additive-error solution. We know
$\Obj(x^a) \leq \Obj(x^*) + \frac{\epsilon}{1+\epsilon}P$ and $\Obj(x^a) \geq P$ where $P=\bnk \delta^2 / 72 k$. Therefore $\Obj(x^*) \geq P (1 - \frac{\epsilon}{1+\epsilon}) = \frac{P}{1+\epsilon}$ and hence $\Obj(x^a) \leq \Obj(x^*) + \frac{\epsilon}{1+\epsilon} (1+\epsilon) \Obj(x^*) = (1+\epsilon) \Obj(x^*)$. Therefore if the additive approximation is returned it is a $1+\epsilon$-approximation.

The remainder of this section considers the case when Algorithm~\ref{alg:main2} takes the ``else'' branch.
 Define $\gamma$ so that $Obj(x^*) = \gamma \bnk$. We have $\Obj(x^*) \leq \Obj(x^a) < \bnk \delta^2 / 72 k$ so $\gamma \leq \delta^2 / 72 k$. We analyze the $\hat{x^*}$ where we guess $x^*$, that is when $\hat{x^*_v} = x^*_v$ for all $v \in \bigcup_{i=1}^s S_i$. Clearly the overall cost at most the cost of $x^{(3)}$ during the iteration when we guess correctly.

\begin{lemma}\label{lem:allZeroClear}
$b(x^*, v, x^*_v) \leq b(x^*,v,j) $ for all $j \in D$.
\end{lemma}

\begin{proof}
Immediate from definition of $b$ and optimality of $x^*$.
\end{proof}

\begin{lemma}\label{lem:objSumBs} For any assignment $x$,
\[
\Obj(x) = \frac{1}{k}\sum_{v \in V} b(x,v,x_v)
\]
\end{lemma}

\begin{proof}
By definitions,
$$b(x,v, x_v) = \sum_{I \in \svk : v \in I} p_I(R_{v x_v}(x)) = \sum_{I \in \svk : v \in I} p_I(x).$$
Write $\Obj(x)=\sum_{I \in \svk} p_I(x)=\sum_{I \in \svk} p_I(x)  \left[\sum_{v \in I}\frac{1}{k}\right]$ and reorder summations.
\end{proof}

\begin{definition}
We say variable $v$ in assignment $x$ is \emph{corrupted} if $x_v \neq x_v^*$.
\end{definition}

\begin{definition}
Variable $v$ is \emph{clear} if $(x^*, v,x_v^*) < b(x^*,v,j) - \frac{\delta}{3} \bnkmo$ for all $j\neq x_v^*$. A variable is \emph{unclear} if it is not clear.
\end{definition}

Clearness is the analysis analog of the algorithmic notion of clear-cut vertices sketched in Section~\ref{sec:intuition}. Comparing the definition of clearness to Lemma~\ref{lem:allZeroClear} further motivates the terminology ``clear.''

\begin{lemma}\label{lem:fewUnclear}
The number of unclear variables $t$ satisfies
\[
t \leq 3 (n-k+1) \gamma / \delta) \leq \frac{\delta n}{24 k}
.\]
\end{lemma}

\begin{proof}
Let $v$ be unclear and choose $j \neq x_v^*$ minimizing $b(x^*,v,j)$. By unclearness, $b(x^*,v,x_v^*) \geq b(x^*,v,j) - (1/3) \delta \bnkmo$. By fragile-dense,  $b(x^*,v,x_v^*) + b(x^*,v,j) \geq \delta \bnkmo$. Adding these inequalities we see
\begin{eqnarray}
b(x^*,v,x_v^*) &\geq& \frac{1-1/3}{2}\delta \bnkmo =  \frac{1}{3}\delta \bnkmo \label{eqn:unclearPricey}
\end{eqnarray}

By Lemma~\ref{lem:objSumBs} and (\ref{eqn:unclearPricey}),
\[
OPT = \gamma \binom{n}{k} = 1/k \sum_v b(x^*, v, x_v^*) \geq 1/k \sum_{v:unclear}\frac{\delta}{3} \bnkmo =  \frac{\delta}{3k} \bnkmo t
.\]
Therefore $t \leq  \gamma \bnk \frac{3k}{\delta\bnkmo} = \frac{3 \gamma}{\delta}(n-k+1)$.

For the second bound observe $3 n \gamma / \delta \leq \frac{3n}{\delta}\frac{\delta^2}{72k} = \frac{\delta n}{24 k}$.
\end{proof}


\begin{lemma}\label{lem:fewClearCorrupted}
The probability of a fixed clear variable $v$ being corrupted in $x^{(1)}$ is bounded above by $\frac{\delta}{240 k}$.
\end{lemma}

\begin{proof}
First we show that $\hat b(v,i)$ is in fact an unbiased estimator of $b(x^*,v,i)$ for all $i$. By definitions and particular by the assumption that $p_I=0$ when $|I|<k$, we have for any $1 \leq j \leq s$:
\begin{eqnarray*}
\E{p_{S_j \cup \{v\}}(R_{vi}(x^*))} &=& \frac{1}{\bnkmo} \sum_{J \in \subsets{V}{k-1}}  p_{J \cup \{v\}}(R_{vi}(x^*)) \\
 &=&\frac{1}{\bnkmo} \sum_{I \in \subsets{V}{k} : v \in I}  p_{I}(R_{vi}(x^*)) \\
&=&\frac{1}{\bnkmo} b_{vi}(x^*)
\end{eqnarray*}
Therefore $\E{\hat b(v,i)} = s\frac{\bnkmo}{s}\E{p_{S_1 \cup \{v\}}(R_{vi}(x^*))} = b(x^*,v,i)$.

Recall that $0 \leq p_I(x) \leq 1$ by definition of $p$, so by Azuma-Hoeffding,
\begin{eqnarray*}
\pr{\left|\sum_{j=1}^{s} p_{S_j \cup \{v\}}(R_{vi}(x^*))  -  \frac{s}{\bnkmo}b(x^*,v,i)\right| \geq \lambda s} &\leq& 2e^{-2\lambda^2 s}
\end{eqnarray*}
hence\commentw{Check the constants}
\begin{eqnarray*}
\pr{|\hat b(v,i)  -  b(x^*,v,i)| \geq \lambda\bnkmo} &\leq& 2e^{-2\lambda^2 s}
\end{eqnarray*}

Choose $\lambda=\delta/6$ and recall  $s=\frac{18 \log (480 |D| k / \delta)}{\delta^2}$, yielding.
\begin{eqnarray*}
\pr{|\hat b(v,i)  -  b(x^*,v,i)| \geq \frac{\delta}{6} \bnkmo} & \leq & \frac{\delta}{240 |D| k}
\end{eqnarray*}

By clearness we have $b(x^*,v,j) > b(x^*,v,x^*_v) + \delta \bnkmo / 3$ for all $j \neq x_v^*$. Therefore, the probability that $\hat b(v,x^*_v)$ is not the smallest $\hat b(v,j)$ is bounded by $|D|$ times the probability that a particular $\hat b(v,j)$ differs from its mean by at least $\delta \bnkmo / 6$. Therefore $\pr{x_v^1 \neq x_v^*} \leq |D|\frac{\delta}{240 |D| k} = \frac{\delta}{240 k}$.
\end{proof}

Let $E_1$ denote the event that the assignment $x^{(1)}$ has at most $\delta n / 12 k$ corrupted variables.

\begin{lemma}\label{lem:fewCorruptedOne}
Event $E_1$ occurs with probability at least $1 - 1/10$.
\end{lemma}

\begin{proof}
We consider the corrupted clear and unclear variables separately. By Lemma~\ref{lem:fewUnclear}, the number of unclear variables, and hence the number of corrupted unclear variables, is bounded by $\frac{\delta n}{24 k}$.

The expected number of clear corrupted variables can be bounded by $\frac{\delta n}{240 k}$ using Lemma~\ref{lem:fewClearCorrupted}, so by Markov bound the number of clear corrupted variables is less than $\frac{\delta n}{24 k}$ with probability at least 1 - 1/10.

Therefore the total number of corrupted variables is bounded by $\frac{\delta n}{24 k} + \frac{\delta n}{24 k} = \frac{\delta n}{12 k}$ with probability at least 9/10.
\end{proof}

We henceforth assume $E_1$ occurs. The remainder of the analysis is deterministic.

\begin{lemma}\label{lem:bdiff}
For assignments $y$ and $y'$ that differ in the assignment of at most $t$ variables,
for all variables $v$ and values $i$, $|b(y,v,i) - b(y',v,i)| \leq t \binom{n}{k-2}$.
\end{lemma}

\begin{proof}
Clearly $p_I(R_{vi}(y))$ is a function only of the variables in $I$ excluding $v$, so
if $I - \{v\}$ consists of variables $u$ where $y_u=y_u'$, then $p_I(R_{vi}(y)) - p_I(R_{vi}(y'))=0$. Therefore
$b(y,v,i)-b(y',v,i)$ equals the sum, over $I \in \svk$ containing $v$ and at least one variable $u$ other than $v$ where $y_u \neq y_u'$, of $\left[ p_I(R_{vi}(y)) - p_I(R_{vi}(y'))\right]$.
For any $I$, $|p_I(R_{vi}(y)) - p_I(R_{vi}(y'))| \leq 1$, so by the triangle inequality a bound on the number of such sets suffices to bound  $|b(y,v,i) - b(y',v,i)|$. The number of such sets can trivially be bounded above by $t \binom{n}{k-2}$.
\end{proof}

\begin{lemma}\label{lem:TGood}
Let $C=\setof{v \in V}{b(x^{(1)}, v,x^{(2)}_v) < b(x^{(1)},v,j) - \delta \bnkmo / 6 \textrm{ for all }j \neq x^{(2)}_v}$ as defined in Algorithm~\ref{alg:main2}. If $E_1$ then:
\begin{itemize}
\item $x_v^{(2)} = x_v^*$ for all $v \in C$.
\item $|V \sm C| \leq \frac{3 n \gamma}{\delta}$.
\end{itemize}
\end{lemma}

\begin{proof}
Assume $E_1$ occurred. From the definition of corrupted, event $E_1$ and Lemma~\ref{lem:bdiff} for sufficiently large $n$ (so that $\frac{n-k+1}{k-1} \geq \frac{n}{k}$) for any $v,i$:
\begin{eqnarray}
|b(x^{(1)},v,i) - b(x^*,v,i)| &\leq& \frac{\delta n}{12 k} \bnkmt \leq \frac{\delta}{12} \bnkmo \label{eqn:bClose2}
.\end{eqnarray}

For the first, if $v \in C$ then using (\ref{eqn:bClose2})
\begin{eqnarray*}
b(x^*,v,x^{(2)}_v) &\leq& b(x^{(2)},v,x^{(2)}_v) + \frac{\delta}{12} \bnkmo <  b(x^{(2)},v,j) - \frac{\delta}{6} \bnkmo + \frac{\delta}{12} \bnkmo \\
&\leq&  b(x^*,v,j) + \left( -\frac{\delta}{6}+ 2 \frac{\delta}{12} \right) \bnkmo =  b(x^*, v, j)
.\end{eqnarray*}
So by Lemma~\ref{lem:allZeroClear}, $x^*_v = x^{(2)}_v$.

For any $u$ that is clear, using (\ref{eqn:bClose2}) again:
\begin{eqnarray*}
b(x^{(2)}, v, x^*_v) &\leq& b(x^*,v,x^*_v) + \frac{\delta}{12} \bnkmo <  b(x^*,v,j) - \frac{\delta}{3} \bnkmo + \frac{\delta}{12} \bnkmo \\
&\leq&  b(x^{(2)},v,j) + \left(-\frac{\delta}{3}+ 2 \frac{\delta}{12}\right)\bnkmo =  b(x^{(2)},v,j) -\frac{\delta}{6}\bnkmo
.\end{eqnarray*}
so by definition of $C$, $u \in C$. Therefore the conclusion follows from Lemma~\ref{lem:fewUnclear}.
\end{proof}

Now we give the details of the computation of $x^{(3)}$. Let $T= V \sm C$. We call $C$ the \emph{clear-cut} vertices and $T$ the \emph{tricky} vertices. We assume that $|T| \geq k$; if not simply consider every possible assignment to the variables in $T$. With the variables in $C$ fixed, those variables can be substituted into the $p_I$ and eliminated. To restore a uniform arity of $k$ we pad the $p_I$ of arity less than $k$ with irrelevant variables from $T$. To ensure none of the resulting $p_I$ has excessive weight we use a uniform mixture of all possibilities for the padding vertices.

If $y$ is an assignment to the variables in $T$ let $R_{Ty}(x^*) = \piecewise{y_v & \tif v \in T \\ x^*_v & \ow}$, a natural generalization of the $R_{vi}(x)$ notation. For $K \in \binom{T}{k}$ and $y \in D^{|T|}$ define
\[
q_K(y) = \sum_{j = 1}^{k} \sum_{J \in \binom{K}{j}} \sum_{L \in \binom{C}{k-j}} p_{J \cup L}(R_{Ty}(x^{(2)})) \binom{|T| - j}{k - j}^{-1}
\]
It is easy to see that $q_K(y)$ is a function only of $y_v$ for $v \in K$ and is hence a cost function analogous to $p_I$ (though not properly normalized).

\begin{lemma}\label{lem:constrainedObj}
For any $y \in D^{|T|}$ we have
\[
Obj(R_{Ty}(x^{(2)})) = \sum_{K \in \binom{T}{k}} q_K(y) + \sum_{I \in \binom{C}{k}} p_I(x^{(2)})
\]
\end{lemma}
\begin{proof}
Let $x = R_{Ty}(x^{(2)})$.
By definition
\begin{eqnarray}
 \sum_{K \in \binom{T}{k}} q_I(y) & = &  \sum_{K \in \binom{T}{k}} \sum_{j = 1}^{k} \sum_{J \in \binom{K}{j}} \sum_{L \in \binom{C}{k-j}} p_{J \cup L}(x) \binom{|T| - j}{k - j}^{-1} \label{eqn:arityPadding}
\end{eqnarray}
Compare to
\begin{eqnarray}
 Obj(x) - \sum_{I \in \binom{C}{k}} p_I(x^{(2)}) & = &  \sum_{I \in \binom{V}{k} : I \not\subseteq C}  p_{I}(x) \label{eqn:arityPadding2}
\end{eqnarray}

Fix $I \in \binom{V}{k}$ and study the weight of $p_I(x)$ in the right-hand-sides of (\ref{eqn:arityPadding}) and (\ref{eqn:arityPadding2}).
 Note there are unique $j \geq 0$,  $J \in \binom{T}{j}$ and $L \in \binom{C}{k-j}$ such that $I = J \cup L$. If $j=0$ then $p_I(x)$ has weight 0 in (\ref{eqn:arityPadding}) and in (\ref{eqn:arityPadding2}). If $j \geq 1$ then $p_I(x)$ appears once in (\ref{eqn:arityPadding}) for each $K \in \binom{T}{k}$ such that $K \supseteq J$. There are $\binom{|T| - j}{k - j}$ of those and each has weight  $\binom{|T| - j}{k - j}^{-1}$ so $p_I(x)$ has an overall weight of 1 in (\ref{eqn:arityPadding}). Clearly $j \geq 1$ implies $I \not\subseteq C$ hence the weight of $p_I(x)$ in (\ref{eqn:arityPadding2}) is 1 as well.
\end{proof}

\begin{lemma}\label{lem:qBound}
\[
0 \leq q_K(y) \leq O\left(\left(\frac{|C|}{|T|}\right)^{k-1}\right)
\]
\end{lemma}
\begin{proof}
Recalling that $0 \leq p_I(y) \leq 1$ and $k=O(1)$:
\[
 q_K(y) \leq \sum_{j=1}^k \binom{k}{j} \binom{|C|}{k-j} \cdot 1 \cdot \binom{|T|-j}{k - j}^{-1} = \sum_{j=1}^k \bigO{\frac{|C|^{k-j}}{|T|^{k-j}} } = \bigO{\frac{|C|^{k-1}}{|T|^{k-1}}}
\]
\end{proof}

Lemma~\ref{lem:qBound} and Theorem~\ref{thm:MS08} with an error parameter of $\epsilon'=\Theta(\epsilon)$ yields an additive error of $O(\epsilon |T|^k (|C|/|T|)^{k-1}) = O(\epsilon (|T|/|C|) n^{k})$ for the problem of minimizing $\sum_{K \in \binom{T}{k}} q_K(y)$. Using Lemma~\ref{lem:TGood} we further bound the additive error $O(\epsilon (|T|/|C|) n^{k})$ by $O(\epsilon \gamma n^k)$. By Lemma~\ref{lem:constrainedObj} this is also an additive error $O(\epsilon \gamma n^k)$ for $Obj(R_{Ty}(x^{(2)}))$.  Lemma~\ref{lem:TGood} implies that $x^* = R_{Ty}(x^{(2)})$ for some $y$, so this yields an additive error  $O(\epsilon \gamma n^k) = \epsilon OPT$ for our original problem of minimizing $Obj(x)$ over all assignments $x$.

\section{Correlation Clustering and Hierarchical Clustering Algorithm}\label{sec:cc}
\subsection{Intuition}

\begin{figure}[tb]
\begin{center}
\includegraphics[width=\textwidth]{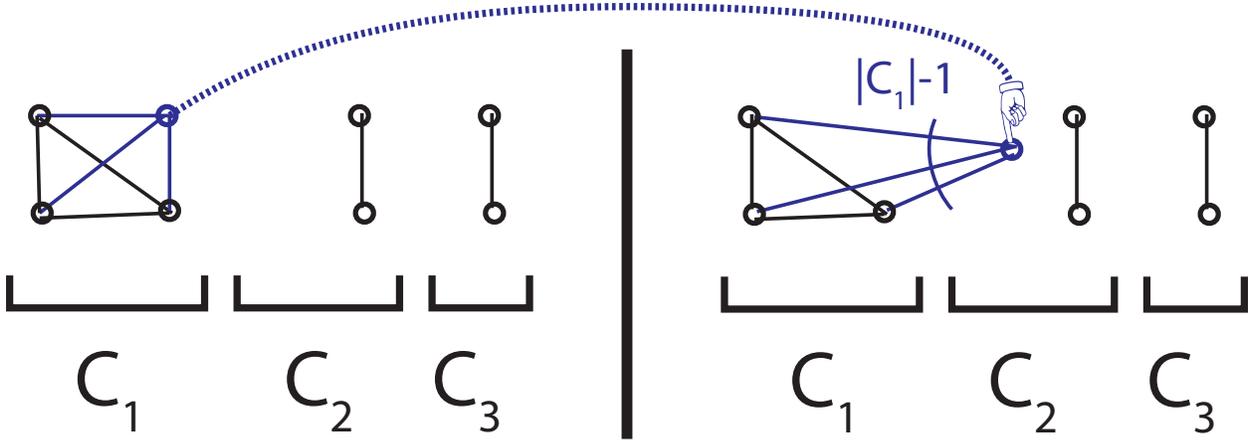}
\end{center}
\caption{An illustration of correlation clustering and the rigidity property.}
\label{fig:rigid}
\end{figure}

As we noted previously in Section~\ref{sec:intro}, correlation clustering constraints are not fragile for $d>2$. Indeed, the constraint corresponding to a pair of vertices that are not connected by an edge can be satisfied by any coloring of the endpoints as long as the endpoints are colored differently. Fortunately there is a key observation in \cite{GG06} that allows for the construction of a PTAS. Consider the cost-zero clustering shown on the left of Figure~\ref{fig:rigid}. Note that moving a vertex from a small cluster to another small one increases the cost very little, but moving a vertex from a large cluster to anywhere else increases the cost a lot. Fortunately most vertices are in big clusters so, as in \cite{GG06}, we can postpone processing the vertices in small clusters. We use the above ideas, which are due to~\cite{GG06}, the fragile-dense ideas sketched above, plus some additional ideas, to analyze our correlation clustering algorithm.

To handle hierarchical clustering (c.f. \cite{AC05}) we need a few more ideas. Firstly we abstract the arguments of the previous paragraph to a CSP property \emph{rigidity}. Secondly, we note that the number of trees with $d$ leaves is a constant and therefore we can safely try them all. We remark that all fragile-dense problems are also rigid.

\subsection{Reduction to Rigid MIN-2CSP}

We now define hierarchical clustering formally (following~\cite{AC05}). For integer $M \geq 1$, an \emph{$M$-level hierarchical clustering} of $n$ objects $V$ is a rooted tree with the elements of $V$ as the leaves and every leaf at depth (distance to root) exactly $M+1$. For $M=1$, a hierarchical clustering has one node at the root, some ``cluster'' nodes in the middle level and all of $X$ in the bottom level. The nodes in the middle level can be identified with clusters of $V$. We call the subtree induced by the internal nodes of a $M$-level hierarchical clustering the \emph{trunk}. We call the leaves of the trunk \emph{clusters}.  A hierarchical clustering is completely specified by its trunk and the parent cluster of each leaf. 

For a fixed hierarchical clustering and clusters $i$ and $j$, let $f(i, j)$ be the distance from $i$ (or $j$) to the lowest common ancestor of $i$ and $j$. For example when $M=1$, $f(i,j)=\one{i=j}$.

We are given a function $F$ from pairs of vertices to $\set{0,1,...M}$.\footnote{\cite{AC05} chose $\set{1,2,,...M+1}$ instead; the difference is merely notational.} The objective of hierarchical clustering is to output a $M$-level hierarchical clustering minimizing $\sum_{u,v} \frac{1}{M} |F(u,v) - f(parent(u),parent(v))|$. Hierarchical clustering with $d$ clusters is the same except that we restrict the number of clusters (recall that equals number of nodes whose children are leaves) to at most $d$. 
The special case of hierarchical clustering with $M=1$ is also called \emph{correlation clustering}.

\begin{lemma}\label{lem:trunk}
The number of possible trunks is at most $d^{(M-1)d}$.
\end{lemma}
\begin{proof}
The trunk can be specified by giving the parent of all non-root nodes. There are at most $d$ nodes on each of the $M-1$ non-root levels so the lemma follows.
\end{proof}

We now show how to reduce hierarchical clustering with a constant number of clusters to the solution of a constant number of min-2CSPs. We use notation similar to, but not identical to, the notation used in Sections~\ref{sec:FD} and~\ref{sec:FDAnalysis}. For vertices $u,v$ and values $i,j$, let $p_{u,v}(i,j)$ be the cost of putting $u$ in cluster $i$ and $v$ in cluster $j$. This is the same concept as $p_I$ for the fragile case, but this notation is more convenient here. Define $b(x,v,i) = \sum_{u \in V, u \neq v} p_{u,v}(x_u, i)$, which is identical to $b$ of the fragile-dense analysis but expressed using different notation.

\begin{definition}
A MIN-2CSP is \emph{rigid} if for some $\delta>0$, all $v \in V$ and all $j \neq x^*_v$
\[
b(x^*,v,x_v^*) + b(x^*,v,j) \geq \delta |\setof{u \in V}{x^*_u = x^*_v}|
\]
\end{definition}

Observe that  $|\setof{u \in V}{x^*_u = x^*_v}| \leq |V|=\binom{|V|}{2-1}$ hence any fragile-dense CSP is also rigid.

\begin{lemma}\label{lem:ccCSP}
If the trunk is fixed, hierarchical clustering can be expressed as a $1/M$-rigid MIN-2CSP with $|D|=d$.
\end{lemma}
\begin{proof}
(C.f. Figure~\ref{fig:rigid})
Choose $\delta=1/M$. Let $D$ be the leaves of the trunk (clusters). It is easy to see that choosing
\[
p_{u,v}(i,j) = \frac{1}{M}|f(i,j) - F(u,v)|
\]
yields the correct objective function. To show rigidity, fix vertex $v$, define $i=x^*_v$ and $\calC_i = \setof{u \in V}{x_u^* = i}$. Fix $j \neq i$ and $u \in \calC_i \sm \set{v}$. Clearly $|f(i,i) - f(i,j)| \geq 1$, hence by triangle inequality $|F(u,v) - f(i,i)| + |F(u,v)-f(i,j)| \geq 1$, hence $p_{u,v}(i,i)+p_{u,v}(i,j) \geq 1/M$. Summing over $u \in \calC_i$ we see
\[
b(x^*,v,x_v^*) + b(x^*,v,j) \geq \frac{1}{M} |\calC_i \sm \set{v}| \approx \frac{1}{M} |\calC_i| = \delta |\setof{u \in V}{x^*_u = x^*_v}|
\]
Sweeping the ``$\approx$'' under the rug this proves the Lemma.\footnote{There are inelegant ways to remove this approximation. For example, assume that all $d$ clusters of $x^*$ are non-empty and consider one vertex from $\calC_j$ as well.}
\end{proof}

Lemmas~\ref{lem:ccCSP} and \ref{lem:trunk} suggest a technique for solving hierarchical clustering: guess the trunk and then solve the rigid MIN-2CSP.
We now give our algorithm for solving rigid MIN-2CSPs.

\subsection{Algorithm for Rigid MIN-2CSP}

Algorithm~\ref{alg:cc} solves rigid MIN-2CSPs by identifying clear-cut variables, fixing their value, and then recursing on the remaining ``tricky'' variables $T$. The recursion terminates when the remaining subproblem is sufficiently expensive for an additive approximation to suffice.

\begin{algorithm}[tb]
\caption{Approximation Algorithm for Rigid MIN-2CSPs.}
\label{alg:cc}
Return CC($V$, blank assignment, 0)

\bigskip

CC(tricky vertices $T$, assignment $y$ of $V \sm T$, recursion depth $depth$):

\begin{algorithmic}[1] 
  \STATE Find assignment of cost at most $\frac{\epsilon}{1+\epsilon} \cdot \frac{ \delta^3 |T|^{2}}{6 \cdot 72^2 |D|^3} + \min_{x: x_v = y_v \forall v \in V \sm T} \left[\Obj(x)\right]$ using an additive approximation algorithm.
\IF {$Obj(answer) \geq \frac{\delta^3 |T|^2}{6 \cdot 72^2 |D|^3}$ or $depth \geq |D|+1$}
 \STATE Return $answer$.
\ELSE
 \STATE Let $s = \frac{432^2 |D|^4 \log (1440 |D|^3/\delta)}{2\delta^4}$
 \STATE Draw $v_1, v_2,\ldots,v_s$ randomly from $T$ with replacement.
 \FOR{Each assignment $\hat{x^*}$ of the variables $\set{v_1, v_2, \ldots, v_s}$}
  \STATE For all $v \in T$ and $i$ let $\hat b(v,i) = \frac{|T|}{s} \sum_{j = 1}^s p_{v_j, v} (\hat{x^*_{v_j}}, i) + \sum_{u \in V \sm T} p_{u,v}(y_u, i)$
  \STATE For all $v \in V$ let $x^{(1)}_v = \piecewise{y_v & \tif v \in V \sm T \\ \arg \min_i \hat b(v,i) & \ow }$
  \STATE For all $v \in T$ let $x^{(2)}_v = \arg \min_i b(x^{(1)},v,i)$
  \STATE Let $C = \setof{v \in T}{b(x^{(1)},v,x^{(2)}_v) < b(x^{(1)},v,j) - \frac{\delta |T|}{12|D|} \textrm{ for all }j \neq x^{(2)}_v} $.
  \STATE Let $T' = T \sm C$
  \STATE Define assignment $y'$ by $y'_v=\piecewise{y_v & \tif v \in V \sm T \\ x^{(2)}_v & \tif v \in C \\ \textrm{Undefined} & \tif v \in T \sm C}$.
  \STATE If CC($T', y', depth+1$) is the best clustering so far, update best.
 \ENDFOR 
 \STATE Return the best clustering found.
\ENDIF
\end{algorithmic}
\end{algorithm}

\section{Analysis of Algorithm~\ref{alg:cc}}\label{sec:ccAnalysis}

\subsection{Runtime}

\begin{theorem}\label{thm:MSVariant}
For any $T,y$, an assignment of cost at most $\epsilon'|T|^2 + \min_{x: x_v = y_v \forall v \in V \sm T} \left[\Obj(x)\right]$ can be found in time $n^2 2^{O(1/\epsilon'^2)}$.
\end{theorem}
\begin{proof}
The problem is essentially a CSP on $T$ vertices but with an additional linear cost term for each vertex. It is fairly easy to see that Algorithm 1 from Mathieu and Schudy~\cite{MS08} has error proportional to the misestimation of $b$ and hence is unaffected by arbitrarily large linear cost terms. On the other hand, the more efficient Algorithm 2 from~\cite{MS08} needs to estimate the objective value from a constant-sized sample as well and hence does not seem to work for this type of problem.
\end{proof}

In this subsection $O(\cdot)$ hides only absolute constants. Algorithm~\ref{alg:cc} has recursion depth at most $|D|+1$ and branching factor $|D|^s$, so the number of recursive calls is at most $(|D|^s)^{|D|+1} = 2^{s (|D|+1) \log |D|} = 2^{\tilde O(|D|^5 / \delta^4)}$. Each call spends $O(|D| n^2)$ time on miscellaneous tasks such as computing the objective value plus time required to run the additive error algorithm, which is $n^2 2^{O(|D|^6/\epsilon^2 \delta^6)}$ by Theorem~\ref{thm:MSVariant}. Therefore the runtime of Algorithm~\ref{alg:cc} is $n^2 2^{O(\frac{|D|^6}{\epsilon^2 \delta^6})}$, where the $2^{\tilde O(|D|^5/\delta^4)}$ from the size of the recursion tree got absorbed into the $2^{O(\frac{|D|^6}{\epsilon^2 \delta^6})}$ from Theorem~\ref{thm:MSVariant}. For hierarchical clustering, $\delta=1/M$ yields a runtime of $n^2 2^{O(\frac{|D|^6 M^6}{\epsilon^2})} \cdot |D|^{(M-1)|D|} = n^2 2^{O(\frac{|D|^6 M^6}{\epsilon^2})}$.

As noted in the introduction this improves on the runtime of $n^{\bigO{\frac{9^{|D|}}{\epsilon^2}}}$ of \cite{GG06} for correlation clustering in two ways: the degree of the polynomial is independent of $\epsilon$ and $|D|$, and the dependence on $|D|$ is singly rather than doubly exponential.

\subsection{Approximation}

We fix optimal assignment $x^*$. We analyze the path through the recursion tree where we always guess $\hat x^*$ correctly, i.e. $\hat x^*_v = x^*_v$ for all $v \in \set{v_1, v_2, \ldots, v_s}$. We call this the \emph{principal path}.

We will need the following definitions.

\begin{definition}
Vertex $v$ is $m$-\emph{clear} if $b(x^*, v, x_v^*) < b(x^*,v,j) - m$ for all $j\neq x_v^*$. We say a vertex is \emph{clear} if it is $m$-clear for $m$ obvious from context. A vertex is \emph{unclear} if it is not clear.
\end{definition}

\begin{definition}
A vertex is \emph{obvious} if it is in cluster $\calC$ in OPT and it is $\delta|\calC|/3$-clear.
\end{definition}

\begin{definition}
A cluster $\calC$ of OPT is \emph{finished} w.r.t. $T$ if $T \cap \calC$ contains no obvious vertices.
\end{definition}

\begin{lemma}\label{lem:ccInduct}
With probability at least 8/10, for any $(T, y, depth)$ encountered on the principle path,
\begin{enumerate}
\item $y_v = x^*_v$ for all $v \in V \sm T$ and
\item The number of finished clusters w.r.t. $T$ is at least $depth$.
\end{enumerate}
\end{lemma}

Before proving Lemma~\ref{lem:ccInduct} let us see why it implies Algorithm~\ref{alg:cc} has the correct approximation factor.

\begin{proof}
Study the final call on the principal path, which returns the additive approximation clustering. The second part of Lemma~\ref{lem:ccInduct} implies that $depth \leq |D|$, hence we must have terminated because $Obj(answer) \geq \frac{\delta^3 |T|^2}{6 \cdot 72^2 |D|^3}$. By the first part of Lemma~\ref{lem:ccInduct} the additive approximation gives error at most
\[
\frac{\epsilon}{1+\epsilon} \cdot \frac{\delta^3 |T|^2}{6 \cdot 72^2 |D|^3} + OPT
.\]
so the approximation factor follows from an easy calculation.
\end{proof}

Now we prove Lemma~\ref{lem:ccInduct} by induction. Our base case is the \emph{root}, which vacuously satisfies the inductive hypothesis since $V \sm T = \set{}$ and $depth=0$. We show that if a node $(T,y, depth)$ (in the recursion tree) satisfies the invariant then its \emph{child} $(T',y', depth+1)$ does as well. We hereafter analyze a particular $(T,y, depth)$ and assume the inductive hypothesis holds for them. There is only something to prove if a child exists, so we hereafter assume the additive error answer is \emph{not} returned from this node. We now prove a number of Lemmas in this context, from which the fact that $T',y',depth+1$ satisfies the inductive hypothesis will trivially follow.

\begin{lemma}\label{lem:ccProb}
The number of $\delta^2 |T|/216D^2$-clear variables that are corrupted in $x^{(1)}$ is at most $\delta |T|/72|D|$ with probability at least $1 - 1/10|D|$.
\end{lemma}
\begin{proof}
Essentially the same proof as for fragile MIN-$k$CSP, and the recursion invariant, shows $\hat b(v,i)$ is an unbiased estimator of $b(x^*,v,i)$.

This time Azuma-Hoeffding yields
\begin{eqnarray*}
\pr{|\hat b(v,i)  -  b(x^*,v,i)| \geq \lambda |T|} &\leq& 2e^{-2\lambda^2 s}
\end{eqnarray*}

Choose $\lambda=\frac{\delta^2}{432 |D|^2}$ and recall  $s=\frac{432^2 |D|^4 \log (1440 |D|^3 / \delta)}{2\delta^4}$, yielding.
\begin{eqnarray*}
\pr{|\hat b(v,i) - b(x^*,v,i)| \geq \delta^2|T|/432|D|^2} & \leq & \frac{\delta}{720 |D|^3}
\end{eqnarray*}

By clearness we have $b(x^*,v,j) > b(x^*,v,x^*_v) + \delta^2|T|/216|D|^2$ for all $j \neq x_v^*$. Therefore, the probability that $\hat b(v,x_v)$ is not the smallest $\hat b(v,j)$ is bounded by $|D|$ times the probability that a particular $\hat b(v,j)$ differs from its mean by at least $\delta^2|T|/432|D|^2$. Therefore $\pr{x_v^{(1)} \neq x_v^*} \leq |D| \frac{\delta}{720 |D|^3} = \frac{\delta}{720 |D|^2}$. Therefore, by Markov bound, with probability $1-1/10|D|$ the number of corrupted $\delta^2|T|/216D^2$-clear variables is at most $\delta |T|/72|D|$.
\end{proof}

There are two types of bad events: the additive error algorithm failing and our own random samples failing. We choose constants so that each of these events has probability at most $1/10|D|$.  This path has length at most $|D|$, so the overall probability of a bad event is at most $2/10$. We hereafter assume no bad events occur.

\begin{lemma}\label{lem:fewPricey}
The number of $\delta c/3$-unclear variables in clusters of size at least $c$ is at most $\frac{6 OPT}{\delta c}$.
\end{lemma}
Let \emph{confusing} variable refer to a $\delta c/3$-unclear variable in a cluster of size at least $c$. Let $v$ be such a variable, in cluster $\calC$ in OPT. By unclearness,
\begin{proof}
\[
b(x^*,v,x_v^*) \geq b(x^*,v,j) - \delta c/3
\]
for appropriate $j \neq x_v^*$ and by rigidity
\[
b(x^*,v,x_v^*) +  b(x^*,v,j) \geq \delta |\calC|
.\]
Adding these inequalities we see $b(x^*,v,x_v^*) \geq  \delta c/3$.

$OPT = 1/2\sum_v b(x^*,v,x_v^*) \geq 1/2\sum_{v\textrm{ confusing}} \delta c/3 = |\setof{v\in T}{v \textrm{ confusing}}| \delta c/6$
so
$|\setof{v\in T}{v \textrm{ confusing}}| \leq \frac{6 OPT}{\delta c}$.
\end{proof}

\begin{lemma}\label{lem:ccBErr} For all $v,i$,
$|b(x^{(1)},v,i) - b(x^*,v,i)| \leq \frac{\delta |T|}{24|D|}$
\end{lemma}
\begin{proof}
First we show bounds on three classes of corrupted variables:
\begin{enumerate}
\item The number of $\delta^2 |T|/216 D^2$-clear corrupted vertices is bounded by $\delta |T|/72|D|$ using Lemma~\ref{lem:ccProb}
\item The number of vertices in clusters of size at most $\delta|T|/72|D|^2$ is bounded by $\delta |T|/72|D|$.
\item The number of $\delta^2 |T|/216 D^2$-unclear corrupted vertices in clusters of size at least $\delta |T|/72|D|^2$ is bounded by, using Lemma~\ref{lem:fewPricey}, $\frac{6OPT}{\delta} \frac{72 |D|^2}{\delta |T|} \leq \frac{\delta^3|T|^2}{6 \cdot 72^2 |D|^3} \cdot\frac{6 \cdot 72 |D|^2}{\delta^2|T|} = \frac{\delta |T|}{72|D|}$.
\end{enumerate}
Therefore the total number of corrupted variables in $x^{(1)}$ is at most $\frac{\delta |T|}{72|D|} + \frac{\delta |T|}{72|D|} + \frac{\delta |T|}{72|D|} = \frac{\delta |T|}{24|D|}$. The easy observation that $|b(x^{(1)},v,i) - b(x^*,v,i)|$ is bounded by the number of corrupted variables in $x^{(1)}$ proves the Lemma.
\end{proof}

\begin{lemma}\label{lem:bigObvious}
There exists an obvious vertex in $T$ that is in a cluster of size at least $|T|/2|D|$.
\end{lemma}
\begin{proof}
Simple counting shows there are at most $|T|/2$ vertices of $T$ in clusters of size less than $|T|/2|D|$.

We say a vertex $v$ is \emph{confusing'} if it is non-obvious and its cluster in OPT has size at least $|T|/2|D|$. By Lemma~\ref{lem:fewPricey}
\[
|\setof{v\in T}{v \textrm{ confusing'}}| \leq \frac{12|D|}{\delta|T|} OPT \leq  \frac{12|D|}{\delta |T|} \frac{\delta^3 |T|^2}{6 \cdot 72^2 |D|^3} < |T|/2
\]

Therefore by counting there must be an obvious vertex in a big cluster of OPT.
\end{proof}

\begin{lemma}\label{lem:somebodyFinishes}
The number of finished clusters w.r.t. $T'$ strictly exceeds the number of finished clusters w.r.t. $T$.
\end{lemma}
\begin{proof}
Let $v$ be the vertex promised by Lemma~\ref{lem:bigObvious} and $\calC_i$ its cluster in OPT. For any obvious vertex $u$ in $\calC_i$ note that $u$ is $\delta |\calC_i|/3 \geq \delta |T|/6|D|$-clear, so Lemma~\ref{lem:ccBErr} implies
\begin{eqnarray*}
b(x^{(1)},u,i) &\leq& b(x^*,u,i) +  \frac{\delta |T|}{24|D|} < b(x^*,u,j) +  \frac{\delta |T|}{24|D|} - \frac{\delta |T|}{6|D|} \\
&\leq&  b(x^{(1)},u,j) +  2\frac{\delta |T|}{24|D|} - \frac{\delta |T|}{6|D|} =  b(x^{(1)},u,j) - \frac{\delta |T|}{12|D|}
\end{eqnarray*}
hence $u \in C$. Therefore, no obvious vertices in $\calC_i$ are in $T'$ so $\calC_i$ is finished w.r.t.
 $T'$. The existence of $v$ implies $\calC_i$ is not finished w.r.t. $T$, so $\calC_i$ is newly finished.
To complete the proof note that $T' \subseteq T$ so finished is a monotonic property.
\end{proof}

\begin{lemma}\label{lem:noMistakes}
$(T',y')$ satisfy the invariant $v \in V \sm T' \to y'_v = x^*_v$.
\end{lemma}
\begin{proof}
Fix $v \in V \sm T'$. If $v \in T$ the conclusion follows from the invariant for $(T,y)$.
If $v \in T \sm T'=C$ we need to show $y'_v = x_v^*$.

Let $i=y'_v$. For any $j \neq i$, use Lemma~\ref{lem:ccBErr} to obtain
\[
b(x^*,v,i) \leq b(x^{(1)},v,i) +  \frac{\delta |T|}{24|D|} <    b(x^{(1)},v,j) +  \frac{\delta |T|}{24|D|} - \frac{\delta |T|}{12|D|} \leq  b(x^*,v,j) +  2\frac{\delta |T|}{24|D|} - \frac{\delta |T|}{12|D|}
=b(x^*,v,j)
\]
so by optimality of $x^*$ we have the Lemma.
\end{proof}

Lemmas~\ref{lem:somebodyFinishes} and \ref{lem:noMistakes} complete the inductive proof of Lemma~\ref{lem:ccInduct}.

\section*{Acknowledgements}

We would like to thank Claire Mathieu and Joel Spencer for raising a question on
approximability status of the Gale-Berlekamp game and Alex Samorodintsky for interesting discussions.

\bibliographystyle{abbrv}
\bibliography{csp}

\begin{thebibliography}{10}

\bibitem{AC05}
N.~Ailon and M.~Charikar.
\newblock Fitting tree metrics: Hierarchical clustering and phylogeny.
\newblock In {\em Procs. 46th IEEE FOCS}, pages 73--82, 2005.

\bibitem{AFKK02}
N.~Alon, W.~{Fernandez de la Vega}, R.~Kannan, and M.~Karpinski.
\newblock Random {S}ampling and {A}pproximation of {MAX-CSP} {P}roblems.
\newblock In {\em 34th ACM STOC}, pages 232--239, 2002.
\newblock journal version in J. Comput. System Sciences 67 (2003), pp. 212-243.

\bibitem{APY08}
N.~Alon, R.~Panigrahy, and S.~Yekhanin.
\newblock Deterministic {A}pproximation {A}lgorithms for the {N}earest
  {C}odeword {P}roblem.
\newblock Technical report, Elec. Coll. on Comp. Compl., ECCC TR08-065, 2008.

\bibitem{ABSS93}
S.~Arora, L.~Babai, J.~Stern, and Z.~Sweedyk.
\newblock The {H}ardness of {A}pproximate {O}ptima in {L}attices, {C}odes, and
  {S}ystems of {L}inear {E}quations.
\newblock In {\em Foundations of Computer Science}, pages 724--733, Nov 1993.

\bibitem{AFK96}
S.~Arora, A.~Frieze, and H.~Kaplan.
\newblock A {N}ew {R}ounding {P}rocedure for the {A}ssignment {P}roblem with
  {A}pplications to {D}ense {G}raph {A}rrangement {P}roblems.
\newblock In {\em Foundations of Computer Science}, pages 21--30, Oct 1996.

\bibitem{AKK95}
S.~Arora, D.~Karger, and M.~Karpinski.
\newblock Polynomial {T}ime {A}pproximation {S}chemes for {D}ense {I}nstances
  of {NP}-{H}ard {P}roblems.
\newblock In {\em 27th ACM STOC}, pages 284--293, 1995.
\newblock journal version in J. Comput. System Sciences 58 (1999), pp. 193-210.

\bibitem{BFK03}
C.~Bazgan, W.~{Fernandez de la Vega}, and M.~Karpinski.
\newblock Polynomial {T}ime {A}pproximation {S}chemes for {D}ense {I}nstances
  of the {M}inimum {C}onstraint {S}atisfaction {P}roblem.
\newblock {\em Random Structures and Algorithms}, 23(1):73--91, 2003.

\bibitem{BK02a}
P.~Berman and M.~Karpinski.
\newblock Approximating {M}inimum {U}nsatisfiability of {L}inear {E}quations.
\newblock In {\em Procs. 13th ACM-SIAM SODA}, pages 514--516, 2002.

\bibitem{BK02b}
P.~Berman and M.~Karpinski.
\newblock Approximation {H}ardness of {B}ounded {D}egree {MIN}-{CSP} and
  {MIN}-{BISECTION}.
\newblock In {\em Procs. 29th ICALP, LNCS 2380}, pages 623--632. Springer,
  2002.

\bibitem{CS04}
J.~Carlson and D.~Stolarski.
\newblock The {C}orrect {S}olution to {B}erlekamp's {S}witching {G}ame.
\newblock {\em Discrete Mathematics}, 287(1--3):145--150, 2004.

\bibitem{DKRS03}
I.~Dinur, G.~Kindler, R.~Raz, and S.~Safra.
\newblock Approximating {CVP} to {W}ithin {A}lmost-{P}olynomial {F}actors is
  {NP}-{H}ard.
\newblock {\em Combinatorica}, 23(2):205--243, 2003.

\bibitem{FL92}
U.~Feige and L.~Lovasz.
\newblock Two prover one round proof systems: Their power and their problems.
\newblock In {\em Procs. 24th STOC}, pages 733--741, 1992.

\bibitem{F96}
W.~{Fernandez de la Vega}.
\newblock {MAX-CUT} has a {R}andomized {A}pproximation {S}cheme in {D}ense
  {G}raphs.
\newblock {\em Random Struct. Algorithms}, 8(3):187--198, 1996.

\bibitem{FKK06}
W.~{Fernandez de la Vega}, R.~Kannan, and M.~Karpinski.
\newblock Approximation of {G}lobal {MAX--CSP} {P}roblems.
\newblock Technical Report TR06-124, Electronic Colloquim on Computation
  Complexity, 2006.

\bibitem{FS89}
P.~C. Fishburn and N.~J. Sloane.
\newblock The {S}olution to {B}erlekamp's {S}witching {G}ame.
\newblock {\em Discrete Math.}, 74(3):263--290, 1989.

\bibitem{FK99}
A.~M. Frieze and R.~Kannan.
\newblock Quick {A}pproximation to {M}atrices and {A}pplications.
\newblock {\em Combinatorica}, 19(2):175--220, 1999.

\bibitem{GG06}
I.~Giotis and V.~Guruswami.
\newblock Correlation clustering with a fixed number of clusters.
\newblock {\em Theory of Computing}, 2(1):249--266, 2006.

\bibitem{GT06}
A.~Gupta and K.~Talvar.
\newblock Approximating unique games.
\newblock In {\em Procs. 17th ACM-SIAM SODA}, pages 99--106, 2006.

\bibitem{Lok95}
S.~V. Lokam.
\newblock Spectral {M}ethods for {M}atrix {R}igidity with {A}pplications to
  {S}ize-{D}epth {T}radeoffs and {C}ommunication {C}omplexity.
\newblock In {\em 36th IEEE FOCS}, pages 6--15, 1995.

\bibitem{MS08}
C.~Mathieu and W.~Schudy.
\newblock Yet {A}nother {A}lgorithm for {D}ense {M}ax {C}ut: {G}o {G}reedy.
\newblock In {\em Procs. 19th ACM-SIAM SODA}, pages 176--182, 2008.

\bibitem{RV08}
R.~Roth and K.~Viswanathan.
\newblock On the {H}ardness of {D}ecoding the {G}ale-{B}erlekamp {C}ode.
\newblock {\em IEEE Transactions on Information Theory}, 54(3):1050--1060,
  March 2008.

\bibitem{RV07}
M.~Rudelson and R.~Vershynin.
\newblock Sampling from large matrices: An approach through geometric
  functional analysis.
\newblock {\em J. ACM}, 54(4):21, 2007.

\bibitem{Spe94}
J.~Spencer.
\newblock {\em Ten {L}ectures on the {P}robabilistic {M}ethod}.
\newblock SIAM, second edition, 1994.

\end{thebibliography}



\end{document}